\documentclass[preprint]{iacrtrans}
\usepackage{lipsum} % Example package -- can be removed
\usepackage[utf8]{inputenc}
\usepackage{geometry}
\usepackage{xcolor}
\usepackage{hyperref}
\usepackage{amsmath}
\usepackage{amssymb}
\usepackage{enumerate}
\usepackage{xpatch}
\usepackage{graphicx}
\usepackage{tabularx}
\usepackage{braket}
\usepackage{physics}
\usepackage{amsthm}
\usepackage{dirtytalk}
\usepackage{tikz}
\usepackage{commath}
\usepackage{mathtools}
\usepackage{qcircuit}
\usepackage{algorithm}
\usepackage{algorithmicx}
\usepackage{algpseudocode}
\usepackage{tabto}
\usepackage{pgf-umlsd}

\usetikzlibrary{matrix,shapes,arrows,positioning,chains,arrows.meta}

\tikzstyle{decision} = [diamond, draw, fill=white, text width=4.5em, text badly centered, node distance=2.5cm, inner sep=0pt]
\tikzstyle{block} = [rectangle, draw, fill=white, text width=5em, text centered, rounded corners, minimum height=4em]
\tikzstyle{block_client} = [rectangle, draw, fill=green, text width=5em, text centered, rounded corners, minimum height=4em]
\tikzstyle{line} = [draw, -latex']
\tikzstyle{cloud} = [draw, ellipse,fill=white, node distance=2cm, minimum height=2em]

\newcounter{protocol}
\makeatletter
\newenvironment{protocol}[1][htb]{%
  \let\c@algorithm\c@protocol
  \renewcommand{\ALG@name}{Protocol}% Update algorithm name
  \begin{algorithm}[#1]%
  }{\end{algorithm}
}

\newcounter{functionality}
\makeatletter
\newenvironment{functionality}[1][htb]{%
  \let\c@algorithm\c@functionality
  \renewcommand{\ALG@name}{Functionality}% Update algorithm name
  \begin{algorithm}[#1]%
  }{\end{algorithm}
}

\makeatother

\newcommand{\A}{\mathcal{A}}
\newcommand{\B}{\mathcal{B}}

\newcommand{\E}{\mathcal{E}}

\newcommand{\inp}{\mathrm{in}}
\newcommand{\out}{\mathrm{out}}
\newcommand{\Ss}{\mathcal{S}}

\newcommand{\F}{\mathcal{F}}
\newcommand{\C}{\mathcal{C}}

\newcommand{\U}{\mathcal{U}}
\newcommand{\G}{\mathcal{G}}
\newcommand{\att}{\text{att}}

\newcommand{\Z}{\mathcal{Z}}
\newcommand{\prog}{\texttt{prog}}
\newcommand{\mem}{\texttt{mem}}
\newcommand{\mpk}{\texttt{mpk}}
\newcommand{\msk}{\texttt{msk}}
\newcommand{\getpk}{\texttt{getpk}}
\newcommand{\install}{\texttt{install}}
\newcommand{\resume}{\texttt{resume}}

\newtheorem{hybrid}{Hybrid}

\author{Yao Ma\inst{1,3} \and Elham Kashefi \inst{1,2} \and Myrto Arapinis \inst{2} \and Kaushik Chakraborty \inst{2} \and Marc Kaplan \inst{3}}
\institute{
  Laboratoire d’Informatique de Paris 6 (LIP6), Sorbonne Université, Paris, France \email{yao.ma@lip6.fr}
  \and
  School of Informatics, University of Edinburgh, Edinburgh, UK, \email{ekashefi@inf.ed.ac.uk, marapini@inf.ed.ac.uk, kchakrab@exseed.ed.ac.uk}
  \and
  VeriQloud, Montrouge, France, \email{kaplan@veriqloud.fr}
}
\title{QEnclave - A practical solution for secure quantum cloud computing}

\begin{document}

\maketitle

%%%% 5. KEYWORDS %%%%
\keywords{Secure Enclave \and Abstract Cryptography \and Quantum Cryptography \and Delegated Quantum Computation \and Cloud Platform}

%%%% 6. ABSTRACT %%%%
\begin{abstract}
We introduce a secure hardware device named a QEnclave that can secure the remote execution of quantum operations while only using classical controls. This device extends to quantum computing the classical concept of a secure enclave which  isolates a computation from its environment to provide privacy and tamper-resistance. Remarkably, our QEnclave only performs single-qubit rotations, but can nevertheless be used to secure an arbitrary quantum computation even if the qubit source is controlled by an adversary. More precisely, attaching a QEnclave to a quantum computer, a remote client controlling the QEnclave can securely delegate its computation to the server solely using classical communication.

We investigate the security of our QEnclave by modeling it as an ideal functionality named \emph{Remote State Rotation}. We show that this resource, similar to previously introduced functionality of remote state preparation, allows blind delegated quantum computing with perfect security. Our proof relies on standard tools from delegated quantum computing. Working in the Abstract Cryptography framework, we show a construction of remote state preparation from remote state rotation preserving the security.

An immediate consequence is the weakening of the requirements for blind delegated computation. While previous delegated protocols were relying on a client that can either generate or measure quantum states, we show that this same functionality can be achieved with a client that only transforms quantum states without generating or measuring them.

It is known that blind quantum computing with information security cannot be implemented using only classical communication. Computational assumptions that circumvent this impossibility induce large overheads on the server side that prevent their practical use. Whereas our approach based on QEnclave does not increase the complexity of the problem on the server side while relying on hardware assumptions that are already used in practice for classical computations. It hence provides for the first time a feasible path towards achieving secure delegated quantum computing for various hardware platforms that are already providing their services over the cloud.  

 % \lipsum[8]
\end{abstract}

%%%% 7. PAPER CONTENT %%%%
\section{Introduction}
\label{sec:main}
% Widely used primitives like the AES~\cite{AES} do not have perfect
% security, and can be analysed with linear
% cryptanalysis~\cite{EC:Matsui93}, differential
% cryptanalysis~\cite{JC:BihSha91}, or differential power
% analysis~\cite{C:KocJafJun99}.  We show that the One-Time-Pad is
% unconditionally secure in \autoref{sec:main}.

Quantum computing is an emerging field of computation technology that  promises to produce faster algorithms for solving computational problems~\cite{shor_polynomial-time_1997,aharonov_polynomial_2006}. Many government agencies and large companies like Google, IBM and Amazon are putting efforts on building a programmable quantum device that can outperform existing classical computers~\cite{arute_quantum_2019,fisher_ibm_2019}. Some of them have already managed to develop small-scale quantum computers and provide cloud services allowing users to delegate their quantum computations~\cite{alsina_experimental_2016,devitt_performing_2016,hebenstreit_compressed_2017,wang_16-qubit_2018}. 

Although this form of \emph{delegated quantum computation} (DQC) services are very useful in practice, for education and research for example, running algorithms on untrusted quantum hardware raises important privacy issues. A major challenge of DQC is to ensure the privacy of the client's computation who doesn't have any quantum computation capability. In this paper, we address this issue 
by introducing a new quantum hardware assumption, namely \emph{quantum trusted execution environment} (Quantum TEE) and showing how it can be used to implement privacy-preserving DQC, even with fully classical client. 

In the classical world, a \emph{Trusted Execution Environment} (TEE) is a secure area that executes code in an isolated environment and prevents malicious access from the rest of the device. The standardization was initially proposed by Global Platform to ensure the protection of stored application and data~\cite{globalplatform_tee_2018}. In practice, a TEE is designed to isolate the trusted execution of software layer from the untrusted area, also called \emph{rich execution environment} (REE). 
It is based on a combination of hardware architecture and cryptographic protection. It allows to control the flow of information between applications in multiple environments with different root-of-trust. In more advanced scenarios, TEEs have been used for blockchain~\cite{lind_teechain_2019}, privacy-preserving machine learning~\cite{grover_privado_2019, ohrimenko_oblivious_2016}, or cloud services~\cite{baumann_shielding_2015, schuster_vc3_2015}. 
% Apart from these, we are also wondering if the definition of TEE could be extended as \emph{Quantum TEE} to provide root-of-trust in quantum computation with reasonable design of architecture.

The goal of \emph{delegated computation} is to allow a computationally bounded client to assign some computation to a computationally powerful but untrusted server while maintaining the privacy of data. This is relevant especially in the case of high performance computing in the cloud.
A similar question arises with universal quantum computers becoming available in the near future. Even though we have recently been witnessing spectacular developments, it is expected that scalable quantum computers will remain hard to build and expensive for a long time. It is very likely that they will only be accessible remotely, exactly like supercomputers are nowadays. In this context, DQC enables a client with limited quantum capabilities to delegate a computation to a quantum server while maintaining the correctness and privacy of the computation. 

The first efficient universal protocol for secure (blind) delegated quantum computation was introduced in~\cite{broadbent_universal_2009} see recent reviews for other similar protocols~\cite{fitzsimons2017private, gheorghiu2019verification}. However, these protocols all assume a quantum channel between the client and the server, which for some quantum hardware platform such as superconducting or cold atom qubits, might prove to be impractical at least in near future. For this reason, the construction of efficient, private and secure DQC protocol using only classical communication will be extremely important. Given the impossibility result on achieving informationally secure delegated computing using only classical communication~\cite{Aaronson2019ComplexityTheoreticLO} other assumption has to be considered. Recent breakthroughs based on post-quantum secure trapdoor one-way functions, paved the way for developing entirely new approaches towards fully-classical client protocols for emerging quantum servers~\cite{urmila_qfhe,mahadev_classical_2018,cojocaru_qfactory_2019}. Nevertheless, the challenge for these protocols is the huge server overhead. This is due to the fact that one has to ensure the quantum circuit implementing the required masking protocol based on the learning with error (LWE) encryption~\cite{Regev} remains unhackable both classically and quantumly. That leads to current proposals that require order of $1000$ server qubits for masking a single gate of the target client computation. 

In our work we explore a new approach based on hardware security assumption to derive practical secure DQC protocol with fully classical client setting. We explore the modular approach introduced in~\cite{dunjko_blind_2016} that defines the \emph{Remote State Preparation} (RSP) as the main building blocks for DQC protocol. It worth noting that in~\cite{gheorghiu_vidick2019} an RSP protocol was also proposed using a classical channel between client and server but assuming a new resource called \emph{Measurement Buffer}, which externalize a quantum state measurement from the server side. However, such a resource can not be realized classically as it was proven in~\cite{badertscher_security_2020}. Indeed it is known that it is impossible to construct a composable secure RSP protocol using only a classical channel between the client and the server without any hardware assumption, which confirms our approach to be the only way forward to construct an efficient DQC protocol with a classical client from the RSP module. One could also take a different approach to define a hardware security module that implement securely the measurement buffer (on the server side) and then use the protocol introduced in~\cite{gheorghiu_vidick2019}. However there are two fall-backs for such protocol. First, it is desired that the hardware assumption should be as simple as possible and as we discuss later securing the measuring device leads to an unnecessary complicated architecture. A more severe issues however is that, as mentioned before due to the usage of an LWE-based encryption, the protocol of~\cite{gheorghiu_vidick2019} requires a huge overhead on the server side. 

With these constraints in mind, we introduce our Quantum TEE, called \emph{QEnclave}, as a practical way to make DQC secure with a classical client. Remarkably only one call to our simple hardware module is enough to create one remote blind qubit. Up to our knowledge this is the first time that a realistic construction for a quantum trusted execution environment has been proposed. Our QEnclave only transforms single qubit states without generating or measuring them. Nevertheless, it can be composed with the universal blind quantum computing protocol of~\cite{broadbent_universal_2009} to achieve secure DQC with perfect blindness (assuming minimal hardware assumption) whiles using only classical communication between the client and the server with optimal server overhead. Surprisingly, the blindness of the protocol holds even if the server controls the qubit source.

The contributions of our work are twofold. The first one is the introduction of a new ideal functionality named \emph{Remote State Rotation} (RSR). The only operation performed by this functionality is to rotate a quantum state with arbitrary angles chosen uniformly at random from a fixed set. In practice, RSR reduces the client's quantum technology requirements compare to previously proposed RSP resources usually requiring the client generating or measuring quantum states. This makes this functionality of independent interest for the study of practical quantum cryptographic protocols, specially that we show how to build RSP from RSR in the \emph{Abstract Cryptography} framework~\cite{maurer_abstract_2011}. In combination with previous results on the security of RSP, it implies that a classical client, using RSR, can achieve DQC with perfect blindness solely relying on classical communication even if the source that generates the state is compromised. However our protocol as it stands does not admit verifiability, and we leave it as an open question whether RSR could also provide verifiable DQC without adding any extra assumptions. 

Our second contribution consists in a proposal to build our QEnclave using a standard classical TEE, together with a protection of the flow between TEE and the quantum apparatus which implements the single-qubit rotations. While such a protection might appear like a strong assumption, it is in practice similar to the requirements for standard classical hardware security modules. For example, the FIPS-140 requirements for cryptographic modules grade the security depending on the guarantees they provide against tampering the hardware. For completeness, we also give a concrete practical classical protocol for the establishment of a post-quantum secure channel between the client and the TEE. Since the  TEE is assumed to guarantee the secure processing of the classical secret decryption, it overall leads to a proposed implementation of a blind DQC protocol using a TEE in our QEnclave.

The rest of the paper is organized as follows: in Section~\ref{sec:prel}, we recall the basic concepts and notations used in our work; in Section~\ref{sec:RSP}, we introduce the functionalities for RSP used in our construction and discuss their composable security in the abstract cryptography framework; in Section~\ref{sec:RSR}, we propose a new ideal functionality that models the QEnclave called remote sate rotation, show how to use it to build a blind DQC protocol and prove the security of the construction;
in section~\ref{sec:QEnclave}, we propose a complete specification of the QEnclave based on the use of a classical TEE, with a protocol to construct a blind DQC protocol with QEnclave.
Finally we conclude our paper with a discussion in Section \ref{sec:concl}. In particular, we discuss how our QEnclave could lead to a verifiable UBQC protocol~\cite{fitzsimons_unconditionally_2017}, and other potential applications of our QEnclave.

% \lipsum

\section{Preliminaries}
\label{sec:prel}
\subsection{Trusted Execution Environment}
A TEE is a tamper-resistant processing environment that runs on a kernel~\cite{sabt_trusted_2015} separated from its environment, named the rich execution environment (REE). It guarantees the authenticity of the executed code, the integrity of the run-time states, and the confidentiality of its code and data. It can also provide remote attestations of its trustworthiness to third parties. A TEE should resist against all software attacks as well as physical attacks performed on the main memory of the system. On the one hand, the OS and most of the applications are executed in the REE and as such might be easily tampered by virus, trojans, malware, tools of rooting/reflashing, keystrokes logging, etc. On the other hand, running applications in the TEE is less efficient than on the REE.

There are many ways to implement a TEE in practice~\cite{gonzalez_operating_2015}. The smartcards we use daily are a prototype of TEE with the smartcard itself being the trusted area while the peripherals (e.g., POS terminal) do not need to be trusted~\cite{koemmerling_design_1999}. Smartcards are completely isolated, providing high levels of trust, but are also very limited due to their size.
A second type of familiar TEE is the \emph{trusted platform module} (TPM)~\cite{trusted_computing_group_tpm_2019}. A TPM is a co-processor specialized for cryptograhic tasks, including key generation, encryption, decryption, \emph{etc}. The trusted components should include isolated engines for cryptography (e.g., SHA-1 engine, RSA engine, HMAC engine, etc.) and a random number generator. In addition, a TPM includes an isolated execution engine, platform configuration registers and persistent memory for identification. %Both types of TEE require secure I/O interface. 

Apart from smartcards and TPMs, there exists another type of TEE that consists in the design of processors with different execution environments and allowing inter-communication among environments with flow control (Figure~\ref{fig:TEE}). Intel SGX, for example, allows users to instantiate a secure \emph{enclave} to protect an application~\cite{intel_intel_2016}. The code from outside the enclave cannot alter the application inside the enclave, even if executed with high privileges. Intel SGX also includes security measures such as remote attestation, crypto-based memory protection, sealing, etc.
Another example is ARM TrustZone, which is implemented in most ARM processors nowadays. The system bus with ARM TrustZone uses an extended protected NS bit to distinguish the instructions of the trusted area from the untrusted area~\cite{arm_arm_2009}. ARM TrustedZone can also protect specific peripherals by hiding them from unstruted applications.
\begin{figure}[h!]
        \centering
        \includegraphics[width=8cm]{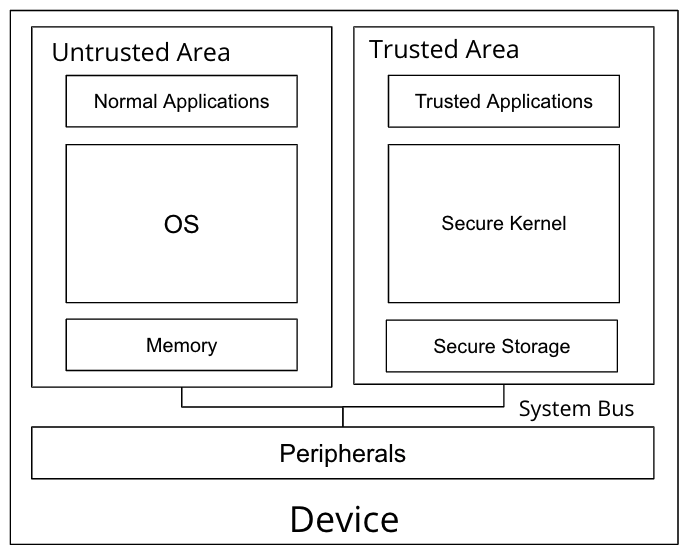}
        \caption{TEE with co-existing execution environments}
        \label{fig:TEE} 
\end{figure}

We introduce a feature that is important for our discussion: \emph{remote attestation}. Remote attestation is a mechanism that allows proving the TEE integrity of a \emph{prover} to a \emph{verifier}. It provides an attestation signed by the TEE manufacturer. For instance, consider a client (verifier) aiming to delegate some application to the TEE of the server (prover). The client can challenge the server to provide him with an attestation signed specifically by the TEE manufacturer and allowing the establishment of an authenticated channel between the client and the TEE, before running a trusted application. The identity and hash of the TEE is a proof of integrity, signed with a hard-coded built-in private key~\cite{sailer_attestation-based_2004}. The proof sent back by the server allows the client to verify the authenticity of the attestation message. Once the attestation is verified, the trusted application runs securely inside the TEE. It also allows anonymous attestation, where a user can verifies an attestation is generated by a valid enclave without identifying which one. The remote attested execution schemes are given in previous works~\cite{barbosa_foundations_2016,coron_formal_2017} to capture the properties of enclave-like secure processors in the real world.

We exploit the abstraction of anonymous attested execution (See Functionality \ref{resource:g_att}) as introduced in~\cite{coron_formal_2017} to formalize cryptographically the secure processors. $\G_\att$ is parameterized by a signature scheme $\Sigma$ and a register \texttt{reg} that captures all parties $\mathcal{P}$ that equip with a secure enclave. For the activation points of $\G_\att$, the starred ones are reentrant activation points, otherwise it can be only executed once. In registry stage, the secure processor enables a distribution of manufacturer public key of key pair $(\mpk, \msk)$ to $\mathcal{P}$ upon query. For enclave operations, the activation point $\texttt{install}$ denotes a new installation of enclave application with a program $\prog$ from $\mathcal{P}$, it generates an identifier \emph{eid} to $\mathcal{P}$ for identifying the enclave instance; the activation point \texttt{resume} enables the execution of $\prog$ upon the input $\texttt{inp}$ by $\G_\att$. $\G_\att$ then signs the output $\texttt{outp}$ to be attested with $\msk$ using $\Sigma$. The attestation $\sigma$ is returned to $\mathcal{P}$ for verification.
\begin{functionality}
    \caption{Anonymous Attested Execution $\G_\att[\Sigma,\texttt{reg}]$}
    \textbf{Registry:}
    
    \vspace{.2cm}
    \quad\textcolor{gray}{// initialization}
        
    \quad On initialize: $(\mpk, \msk):=\Sigma.\texttt{KeyGen}(1^\lambda), T=0$
    
    \vspace{.2cm}
    \quad\textcolor{gray}{// public query interface}
        
    \quad On receive* $\getpk()$ from some $\mathcal{P}$: send $\mpk$ to $\mathcal{P}$
    
    \vspace{.2cm}    
    \textbf{Enclave Operations:}
    
    \vspace{.2cm}
    \quad\textcolor{gray}{//install an enclave program}
    
    \quad On receive* $\install(idx,\prog)$ from some $\mathcal{P}\in \texttt{reg}:$
        
    \quad\quad if $\mathcal{P}$ is honest, assert $idx=sid$
        
    \quad\quad generate nonce $eid\in\{0,1\}^\lambda$, store $T[eid,\mathcal{P}]:=(eid, \prog,0)$, send $eid$ to $\mathcal{P}$
    
    \vspace{.2cm}   
    \quad\textcolor{gray}{//resume an enclave program} 
    
    \quad On receive* $\resume(eid,\texttt{inp})$ from some $\mathcal{P}\in \texttt{reg}:$
    
    \quad\quad let $(idx,\prog,\mem):=T[eid,\mathcal{P}]$, abort if not found
        
    \quad\quad let $(\texttt{outp},\mem):=\prog(\texttt{inp},\mem)$, update $T[eid, \mathcal{P}]:=(idx,\prog,\mem)$
        
    \quad\quad $\sigma:=\Sigma.\texttt{Sig}_\msk(idx,eid,\prog,\texttt{outp})$, and send $(\texttt{outp},\sigma)$ to $\mathcal{P}$
    \label{resource:g_att}
\end{functionality}

\subsection{Quantum Tools}
We introduce the basic concepts required here. Interested readers can refer to standard textbooks on this topic~\cite{nielsen_quantum_2010}.
In quantum computation, a quantum bit or (\emph{qubit}) is a quantum system that is the analogue of a classical bit. It lives in a two-dimensional Hilbert space~$\mathcal{H}$. In particular, the qubits of the \emph{computational basis} of~$\mathcal{H}$ are denoted:
\begin{align*}
    \ket{0}=
    \begin{bmatrix}
    1 \\
    0
    \end{bmatrix},
    \hspace{1cm}
    \ket{1}=
    \begin{bmatrix}
    0 \\
    1
    \end{bmatrix}.
\end{align*}
More generally, the state of an arbitrary qubit is described as \(\ket{\psi}=\alpha\ket{0}+\beta\ket{1}\) where \(\abs{\alpha}^{2}+\abs{\beta}^{2}=1\). 
An alternative basis called the \emph{Hadamard basis} consists in the following qubits:
\begin{align*}
    \ket{+}=\frac{1}{\sqrt{2}}
    \begin{bmatrix}
    1 \\
    1
    \end{bmatrix},
    \hspace{1cm}
    \ket{-}=\frac{1}{\sqrt{2}}
    \begin{bmatrix}
    1 \\
    -1
    \end{bmatrix}.
\end{align*}
%Alternatively, an arbitrary single-qubit quantum state can be described as:
%\begin{align*}
%    \ket{\psi}=cos(\frac{\phi}{2})\ket{0}+e^{i\theta}sin(\frac{\phi}{2})\ket{1},
%\end{align*}
%where \(\theta\in[0,2\pi]\) and \(\phi\in[0,\pi]\). 
%Using $\theta$ and $\phi$, the state can be mapped to a three-dimensional sphere called \emph{Bloch sphere}.
We will in particular make use of the transform $Z(\theta)$ that maps $\ket{\pm}$ to
\(\ket{\pm_{\theta}}=\frac{1}{\sqrt{2}}(\ket{0}\pm e^{i\theta}\ket{1})\). %$Z(\theta)$ corresponds to a rotation of the X-Y plane of Bloch sphere. 
%\begin{figure}[H]
%    \centering
%    \begin{tikzpicture}[line cap=round, line join=round, >=Triangle]
%      \clip(-2.19,-2.49) rectangle (2.66,2.58);
%      \draw [shift={(0,0)}, lightgray, fill, fill opacity=0.1] (0,0) -- (56.7:0.4) arc (56.7:90.:0.4) -- cycle;
%      \draw [shift={(0,0)}, lightgray, fill, fill opacity=0.1] (0,0) -- (-135.7:0.4) arc (-135.7:-33.2:0.4) -- cycle;
%      \draw(0,0) circle (2cm);
%      \draw [rotate around={0.:(0.,0.)},dash pattern=on 3pt off 3pt] (0,0) ellipse (2cm and 0.9cm);
%      \draw (0,0)-- (0.70,1.07);
%      \draw [->] (0,0) -- (0,2);
%      \draw [->] (0,0) -- (-0.81,-0.79);
%      \draw [->] (0,0) -- (2,0);
%      \draw [dotted] (0.7,1)-- (0.7,-0.46);
%      \draw [dotted] (0,0)-- (0.7,-0.46);
%      \draw (-0.08,-0.3) node[anchor=north west] {$\theta$};
%      \draw (0.01,0.9) node[anchor=north west] {$\phi$};
%      \draw (-1.01,-0.72) node[anchor=north west] {$x$};
%      \draw (2.07,0.3) node[anchor=north west] {$y$};
%      \draw (-0.5,2.6) node[anchor=north west] {$z\ket{0}$};
%      \draw (-0.4,-2) node[anchor=north west] {$\ket{1}$};
%      \draw (0.4,1.65) node[anchor=north west] {$\ket{\psi}$};
%      \scriptsize
%      \draw [fill] (0,0) circle (1.5pt);
%      \draw [fill] (0.7,1.1) circle (0.5pt);
%    \end{tikzpicture}
%    \caption{Bloch sphere and an arbitrary state $\ket{\psi}$}
%    \label{fig:Bloch}
%\end{figure}

A quantum state can also be described by its \emph{density matrix} \(\rho=\ketbra{\psi}{\psi}\). Density matrices also capture mixed states of the form $\rho=\Sigma_{s}p_{s}\ketbra{\psi_{s}}{\psi_{s}}$ where $p_{s}$ is a probability over pure states $\ketbra{\psi_{s}}{\psi_{s}}$. 

For multiple qubit systems, two states $\ket{v}$ and $\ket{w}$ in two Hilbert spaces $V$ and $W$ with dimension $n$ and $m$ can be assembled as \(\ket{v}\otimes\ket{w}\), or simply $\ket{vw}$, which lives in
$V\otimes W$, a $n\cdot m$ dimensional Hilbert space. A quantum system $\ket u$ is called \emph{separable} if it can be written $\ket v \otimes \ket w$. A multiple qubit system that is not separable is \emph{entangled}. For example, \(\ket{\phi}=\frac{1}{\sqrt{2}}(\ket{00}+\ket{11})\) is an entangled state.

The measurement of a quantum state is defined by a set of operators $\{M_{i}\}$ satisfying $\sum_i M_{i}^{\dagger}M_{i} = I$ with its conjugate transpose operator $M^{\dagger}$. The probability of getting measurement result $i$ on quantum state $\ket{\psi}$ is:
\begin{align*}
    P(i)=\bra{\psi}M_{i}^{\dagger}M_{i}\ket{\psi}=\bra{\psi}M_{i}\ket{\psi}.
\end{align*}
In particular, if $\mathcal B = \{\ket u, \ket v\}$ is a basis of qubit states, then the measurement defined by the operators $\{\ketbra{u}{u}, \ketbra{v}{v} \}$ is usually referred to as a projection onto basis $\mathcal B$.

The transformation of a quantum state can be described by a \emph{unitary} operator $U$. These transforms preserve the norm of a vector and hence map a quantum state onto another quantum state.
%where it satisfies \(UU^{\dagger}=I\) with its conjugate transpose operator $U^{\dagger}$. One of the most important reason that using unitary operators in quantum computation is that they preserve the norm of a vectors, and hence the normalization of quantum states still holds. Besides, it follows that any non-measuring quantum operation on quantum states must be reversible, where $U^{\dagger}$ could \say{undo} the action of $U$.
%When working with density matrices, 
%\emph{completely positive and trace-preserving} (CPTP) maps are more general transformation of quantum system which in particular can map pure states onto mixed states.
%A CPTP map $\mathcal{E}$ can always be described as a linear combination of \emph{Kraus operators} \(\{E_{k}\}\), where \(\Sigma_{k} E_{k}^{\dagger}E_{k}=I\). (\textcolor{blue}{we mention once CPTP map in the definition of DQC ideal resource})

We use letters $X/Y/Z$ to denote some particular unitary operators called \emph{Pauli operators}. For single-qubit, the Pauli operators, as well as identity $I$, are given in the following matrices:
\begin{align*}
    I = \begin{pmatrix} 1 & 0 \\ 0 & 1 \end{pmatrix},
    X = \begin{pmatrix} 0 & 1 \\ 1 & 0 \end{pmatrix},
    Y = \begin{pmatrix} 0 & i \\ -i & 0 \end{pmatrix}, 
    Z = \begin{pmatrix} 1 & 0 \\ 0 & -1 \end{pmatrix}.
\end{align*}
%Together with the products of $\pm1$ and $\pm i$, we obtain the matrix group so-called $\emph{Pauli group}$:
%\begin{align*}
%    C_{1}=\{\pm I, \pm iI, \pm X, \pm iX, \pm Y, \pm iY, \pm Z, \pm iZ\}.
%\end{align*}
%Notice that Kraus operators $E_{k}$ can be written as \(E_{k}=\Sigma_{i}\alpha_{i}\sigma_{i}\), where $\alpha_{i}$ is a complex number and $\sigma_{i}$ is a Pauli operator~\cite{nielsen_quantum_2010}. 

%Another group of unitary operators is called \emph{Clifford group}, which could leave the Pauli group invariant under \emph{conjugation}~\cite{danos_measurement_2007}.
The other operators relevant here are the \emph{Hadamard} ($H$) gate, which maps the computational basis to the Hadamard basis, and the \emph{control-U} ($CU$) gates, which uses two qubits as input: a control qubit and target qubit. It operates $U$ on the target qubit when the control qubit is set to $\ket 1$. 
%\begin{align*}
%    H = \frac{1}{\sqrt{2}}\begin{pmatrix} 1 & 1 \\ 1 & -1 \end{pmatrix},
%    CNOT = \begin{pmatrix} 1 & 0 & 0 & 0 \\ 0 & 1 & 0 & 0 \\ 0 & 0 & 0 & 1 \\ 0 & 0 & 1 & 0 \end{pmatrix},
%    CZ = \begin{pmatrix} 1 & 0 & 0 & 0 \\ 0 & 1 & 0 & 0 \\ 0 & 0 & 1 & 0 \\ 0 & 0 & 0 & -1 \end{pmatrix},
%    P = \begin{pmatrix} 1 & 0 \\ 0 & i \end{pmatrix}
%\end{align*}

Finally, we give a very brief introduction to a model of quantum computation called measurement-based quantum computation (MBQC), originally proposed by Raussendorf and Briegel~\cite{raussendorf_one-way_2001, raussendorf_measurement-based_2003, briegel_measurement-based_2009}. The DQC protocols discussed in our work are well described in the MBQC computation model. In this model, a computation is described by a series of commands involving single-qubits or two qubits: preparations of single-qubits in the state $\ket{+}$; entanglements of two qubits with the $CZ$ operator; measurements on single-qubits with basis $\ket{+_{\theta}}$ and $\ket{-_{\theta}}$ with measurement results (\emph{signals}) $0$ and $1$ respectively, corrections on single-qubits with operators $X$, $Z$ depending on signals~\cite{danos_measurement_2007}.

The entangled state used for computation in MBQC is called a \emph{graph state}. An MBQC computation is a sequence of commands on a graph state that includes a subset of input qubits and output qubits. In the family of graph state, \emph{cluster states} introduced in~\cite{raussendorf_measurement-based_2003} and \emph{brickwork states} introduced in~\cite{broadbent_universal_2009} are proved to be universal for MBQC.

\subsection{Abstract Cryptography}
The \emph{Abstract cryptography} (AC, also called Constructive Cryptography) framework was introduced in~\cite{maurer_abstract_2011} by Maurer and Renner for getting composable security properties. 
Compared to UC framework~\cite{959888,10.1007/978-3-540-70936-7_4} that is built in bottom-up approach, AC framework is formalized with top-down approach, where it considers the highest level of abstraction first, then the lower levels to instantiate particular objects of protocol. UC can be realized by instantiating the abstraction of AC framework. However, it is not our goal to compare different approaches in this paper, but to show the idea behind composable security.
Both these frameworks provide a method to establish the comparison of similarity between different functionalities, and further define the composable security of these functionalities. 

In AC framework, the functionality is called a \emph{resource}. A resource has some \emph{interfaces} $\mathcal{I}$ corresponding to the parties that the resource interacts with. Since we focus on two-party communication between the client and the server as our protocol, our resources have two interfaces \(\mathcal{I=\{\mathcal{A},\mathcal{B}\}}\) corresponding to the client and the server.

A protocol \(\pi=\{\pi_{i}\}_{i\in\mathcal{I}}\) is a set of \emph{converters} indexed by $\mathcal{I}$. A converter has two interfaces - an inside interface and an outside interface, where the inside interface is connected to the resource and the outside interface is connected to the outside world. 
Intuitively, a dishonest party in a protocol has more access to the functionalities of a resource than an honest one. We denote by $\perp$ a filter used to enforce the honest behaviour of a party. In this case, the functionalities accessed by the party are so-called the \emph{filtered functionalities}. 

An important concept of the AC framework is the \emph{distinguisher} ($\mathcal{D}$), which measures the distance between two resources. For instance, consider a resource $\mathcal{R}$ and a protocol \({\pi_{A},\pi_{B}}\), and denote \(\pi_{A}\mathcal{R}\pi_{B}\) their composition. We say that two resources \({\mathcal{R},\mathcal{S}}\) are $\varepsilon$-closed, or \(\mathcal{R}\approx_{\varepsilon}\mathcal{S}\) if there is no distinguisher $\mathcal{D}$ that can distinguish between $\mathcal{R}$ and $\mathcal{S}$ with advantage greater that $\varepsilon$. If $\varepsilon$ is negligible, we say that we can construct $\mathcal{S}$ from $\mathcal{R}$ with the protocol \({\pi_{A},\pi_{B}}\). Furthermore, if the resource $\mathcal{S}$ is secure, we say that the resource $\mathcal{R}$ securely constructs $\mathcal{S}$.
The following definition formally defines this.

\begin{definition}
\textnormal{(See~\cite{maurer_abstract_2011})} Given two resource $\mathcal{R}$ and $\mathcal{S}$, we say that a protocol \(\pi=\{\pi_{A},\pi_{B}\}\) constructs $\mathcal{S}$ from $\mathcal{R}$ 
within $\varepsilon$ if two following properties are satisfied:
\begin{itemize}
    \item Correctness:\\
    \begin{align}
        \pi_{A}\mathcal{R}\pi_{B}\approx_{\varepsilon}\mathcal{S}\perp,
    \end{align}
    \item Security: if there exists a converter, where it is called a \emph{simulator} $\sigma$ such that\\
    \begin{align}
        \pi_{A}\mathcal{R}\approx_{\varepsilon}\mathcal{S}\sigma.
    \end{align}
\end{itemize}
We denote this:
\begin{align}
    \mathcal{R}\xrightarrow{\pi,\varepsilon}\mathcal{S}
\end{align}
\label{defn:ac} 
\end{definition}

\subsection{Delegated Quantum Computing}
In a client-server  delegated quantum computation (DQC) protocol, a client with limited computational power asks a server to run a quantum computation, whose result is then returned to the client. There exist two types of  DQC protocols. The first ones are \emph{prepare-and-send} protocols, in which the client prepares a certain numbers of quantum states and send them to the server. The second class is \emph{receive-and-measurement} protocols~\cite{hayashi_verifiable_2015}, where the client receives single-qubits from the server and measures them. 

When delegating its computation, a client expects some security guarantees. The first one is blindness, which means that the server does not learn anything about the computation, input and output. The second one is verifiability, which means that a client can verify the correctness of the the result returned by the server.

\subsubsection{Ideal Functionalities of DQC}
The following definition from~\cite{dunjko_composable_2014} specifies an ideal resource for two-party delegated quantum computing with blindness.
\begin{definition}
\textnormal{(See Figure \ref{fig:DQC}(a))} 
For a given unitary $\mathcal U$, the ideal resource for DQC $\mathcal{S}^{blind}$ provides both correctness and blindness. It takes an input $\psi_{A}$ at the client's interface, and on the server's interface, a filtered control bit $c$ (set by default to 0) and a pair that consists in a state $\psi_{B}$ and a description of a CPTP map $\mathcal{E}$ \footnote{A CPTP map $\mathcal{E}$ is a generalization of unitary operators to density matrices. It
can always be described as a linear combination of \emph{Kraus operators} \(\{E_{k}\}\), which can be written as \(E_{k}=\Sigma_{i}\alpha_{i}\sigma_{i}\), where $\alpha_{i}$ is a complex number and $\sigma_{i}$ is a Pauli operator}. It outputs the allowed leak $\ell^{\psi_{A}}$ at the server's interface. If $c=0$, it outputs the correct result $\mathcal{U}(\psi_{A})$ at the client's interface; otherwise it outputs the server's choice, $\mathcal{E}(\psi_{AB})$
\label{defn: BQC}
\end{definition}

The blindness means there is at most $\ell^{\psi_{A}}$ of information leaked to the server during the interactions. The other property of DQC that we are interested in is verifiability. This means that if a dishonest server returns an incorrect result, the probability that the client accepts it is negligible.
The following definition formalizes the definition of verifiable DQC.
\begin{definition}
\textnormal{(See Figure \ref{fig:DQC}(b))} 
For a given unitary $\mathcal U$, the ideal resource DQC resource $\mathcal{S}_{verif}^{blind}$ provides correctness, blindness and verifiability. It takes an input $\psi_{A}$ at the client's interface, and a filtered control bits $c$ (set by default to 0) at the server's interface. It outputs the allowed leak $\ell^{\psi_{A}}$ at the server's interface. If $c=0$, it simply outputs $\mathcal{U}(\psi_{A})$ at the client's interface. If $c=1$, it outputs an error message at the client's interface. 
\label{defn: VBQC}
\end{definition}

% \begin{figure}[h!]
%         \centering
%         \includegraphics[width=13cm]{img/DQC.png}
%         \caption{DQC ideal resources with blindness (a) and DQC ideal resources with both blindness and verifiability (b)}
%         \label{fig:DQC} 
% \end{figure}
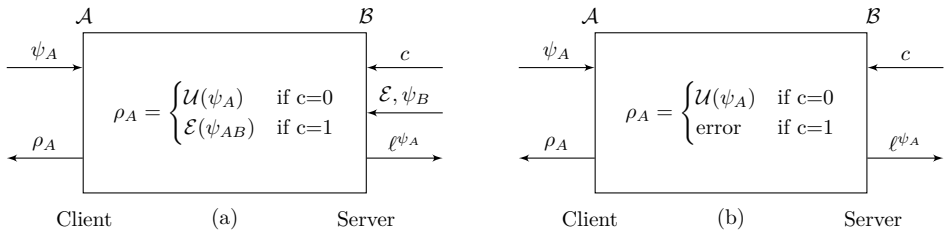
\begin{figure}[h!]
    \centering
    \begin{tikzpicture}[every node/.style = {rectangle, align=center,scale=0.8}, scale=0.8]
    \node(blindness)[draw, inner xsep=5mm, inner ysep=8mm, minimum size=1.5cm]{
    $\rho_A=
    \begin{cases}
    \U(\psi_A)  &\text{if c=0}\\
    \E(\psi_{AB}) &\text{if c=1}
    \end{cases}$};
    \node(veri)[draw, inner xsep=5mm, inner ysep=8mm, minimum size=1.5cm, right=3cm of blindness]{
    $\rho_A=
    \begin{cases}
    \U(\psi_A)  &\text{if c=0}\\
    \text{error} &\text{if c=1}
    \end{cases}$};
    \node(a) [below = 1mm of blindness] {(a)};
    \node(b) [below = 1mm of veri] {(b)};
    \node(client_a) [left = 1.1cm of a] {Client};
    \node(server_a) [right = 1.1cm of a] {Server};
    \node(client_b) [left = 1.1cm of b] {Client};
    \node(server_b) [right = 1.1cm of b] {Server};
    \node(interface_a1) [above = 2.3cm of client_a] {$\A$};
    \node(interface_a2) [above = 2.3cm of server_a] {$\B$};
    \node(interface_b1) [above = 2.3cm of client_b] {$\A$};
    \node(interface_b2) [above = 2.3cm of server_b] {$\B$};
    \coordinate[left = of blindness] (a0);
    \coordinate[above = 6mm of a0] (a1);
    \coordinate[below = 6mm of a0]   (a2);
    \coordinate[right = of blindness] (a3);
    \coordinate[above = 6mm of a3] (a4);
    \coordinate[below = 6mm of a3] (a5);
    \draw[-latex'] (a1) -- node[pos=0.5, above] {$\psi_A$} (a1-| blindness.west);
    \draw[-latex'] (a2-| blindness.west) -- node[pos=0.5, above] {$\rho_A$} (a2);
    \draw[-latex'] (a3) -- node[pos=0.5, above] {$\E,\psi_B$} (a3-| blindness.east);
    \draw[-latex'] (a4) -- node[pos=0.5, above] {$c$}  (a4-| blindness.east);
    \draw[-latex'] (a5-| blindness.east) -- node[pos=0.5, above] {$\ell^{\psi_A}$} (a5);
    \coordinate[left = of veri] (b0);
    \coordinate[above = 6mm of b0] (b1);
    \coordinate[below = 6mm of b0]   (b2);
    \coordinate[right = of veri] (b3);
    \coordinate[above = 6mm of b3] (b4);
    \coordinate[below = 6mm of b3] (b5);
    \draw[-latex'] (b1) -- node[pos=0.5, above] {$\psi_A$} (b1-| veri.west);
    \draw[-latex'] (b2-| veri.west) -- node[pos=0.5, above] {$\rho_A$} (b2);
    \draw[-latex'] (b4) -- node[pos=0.5, above] {$c$}  (b4-| veri.east);
    \draw[-latex'] (b5-| veri.east) -- node[pos=0.5, above] {$\ell^{\psi_A}$} (b5);
    \end{tikzpicture}
    \caption{DQC ideal resources with blindness (a) and DQC ideal resources with both blindness and verifiability (b)}
    \label{fig:DQC} 
\end{figure}

\subsubsection{Universal Blind Quantum Computation}
Universal delegated quantum computation (UBQC), originally introduced in~\cite{broadbent_universal_2009}, is a quantum computation model whose operations can easily be described in the MBQC model. At the start of a UBQC protocol, the client produces a sequence of single-qubit states of the form \(\ket{+_{\theta}}\) with  \(\theta\) chosen uniformly at random from \(\{0,\frac{\pi}{4},\dots\frac{7\pi}{4}\}\). Here and throughout   this paper, we use $\{\mathbb{Z}_{\frac{\pi}{4}}\}$ to denote this set of angles.
After receiving $N$ such qubits from the client through a quantum channel, the server entangles them to build a brickwork state.

The computational stage is interactive, and uses only classical communication. During this stage, the client continuously sends the measurement angle for each qubit to the server, who returns the measurement result to the client. The client then computes the following measurement angle. At the end of the computation, the server returns the quantum outputs to the client. Dunjko, Fitzsimons, Portmann and Renner~\cite{dunjko_composable_2014}  showed the security of a UBQC protocol providing perfect blindness in the AC framework. 

For completeness, we provide a general description of UBQC in Protocol~\ref{protcol:UBQC}. More details can be found in~\cite{broadbent_universal_2009}.
\begin{protocol}[!h]
\caption{Universal Blind Quantum Computation}
\begin{enumerate}
    \item \textbf{Client's preparation:}
    \begin{enumerate}
        \item The client prepares a unitary map $\mathcal{U}$ described as a pattern on a brickwork state. $\mathcal{G}_{n \times m}$.
        \item The client generates $n\times m$ qubits in states \(\ket{\psi_{x,y}}\in_{R}\{\ket{+_{\theta_{x,y}}}=\frac{1}{\sqrt{2}}(\ket{0}+e^{i\theta_{x,y}}\ket{1})\ |\ \theta_{x,y}\in\{\mathbb{Z}_{\frac{\pi}{4}}\}\}\), with measurement angles $\phi_{x,y}$,  dependency sets $X_{x,y}$ and $Z_{x,y}$ obtained from the flow construction~\cite{danos_determinism_2006} corresponding to $\mathcal{U}$. The qubits are sent to the server.
    \end{enumerate}
    \item \textbf{Server's preparation:}
    \begin{enumerate}
        \item The server creates $n$ qubits in the $\ket{+}$ state to use as the final output layer.
        \item The server creates an entangled state using both the qubits received from the client and the output layer qubits by applying ctrl-$Z$ gates between the qubits in order to create a brickwork state $\mathcal{G}_{n \times m+1}$.
    \end{enumerate}
    \item \textbf{Interactions and Measurement:}\\
    For each column \(x=1,\dots,m\)
    
    $\>$ For each row \(y=1,\dots,n\)
    \begin{enumerate}
        \item The client computes the updated measurement angle $\phi_{x,y}^{\prime}$, which depends on previous measurement outcomes reported by the server, and some random choices $r_{x,y}$ made by the client to hide the measurement angles.
        \item The client chooses a binary value \(r_{x,y}\in\{0,1\}\) uniformly at random, and computes \(\delta_{x,y}=\phi_{x,y}+\theta_{x,y}+\pi r_{x,y}\).
        \item The client sends $\delta_{x,y}$ to the server, who performs a measurement in the basis \(\{\ket{+_{\delta_{x,y}}},\ket{-_{\delta_{x,y}}}\}\).
        \item The server sends the result \(s_{x,y}\in\{0,1\}\) to the client.
        \item If \(r_{x,y}=1\), the client flips $s_{x,y}$; otherwise it does nothing.
    \end{enumerate}
    \item \textbf{Output Correction}
    \begin{enumerate}
        \item The server sends to the client all qubits in the last output layer.
        \item The client performs the final Pauli corrections \(\{Z^{s^{Z}_{x,m}},X^{s^{X}_{x,m}}\}^{n}_{x=1}\) on the received output qubits.
    \end{enumerate}
\end{enumerate}
\label{protcol:UBQC}
\end{protocol}

\section{Remote State Preparation for DQC}
\label{sec:RSP}
In this section, we review the works on ideal functionalities RSP and their  security in the AC framework.
Using remote state preparation (RSP) as an ideal functionality allows to replace the quantum channel between a client and a server by a classical one.

\subsection{Remote State Preparation for Blindness}
The UBQC protocol introduced above requires the server to get a number of states of the form $\ketbra{+_{\theta}}{+_{\theta}}$.
These states are then entangled to construct a brickwork state. 
Dunjko and Kashefi~\cite{dunjko_blind_2016} have introduced the concept of \emph{weak correlations}, which is a necessary and sufficient condition on the set of states sent by the client to obtain the blindness of the protocol.
The following theorem formally introduces this notion.
\begin{theorem}
\textnormal{(See~\cite{dunjko_blind_2016})} The UBQC protocol (Protocol \ref{protcol:UBQC}) with classical input and computation of size $N$, where the client's preparation stage is replaced by the preparation of $N$ states of the form 
$\sigma_{AB}^{i}$ 
\begin{align}
    \sigma_{AB}^{i}=\frac{1}{\abs{\Theta}}\Sigma_{\theta_{i}\in\Theta}\ketbra{\theta_{i}}{\theta_{i}}\otimes\rho_{i}^{\theta_{i}},
\label{eqn:wc}
\end{align}
%\footnote{$\ket{\theta}\bra{\theta}$ denotes the classical information of the client's register}:
is blind if and only if the following conditions hold:
\begin{enumerate}
    \item $\rho^{\theta}$ is a normalized quantum state, for all $\theta$,
    \item \(\rho^{\theta}+\rho^{\theta+\pi}=\rho^{\theta^{\prime}}+\rho^{\theta^{\prime}+\pi}\) for all $\theta$, $\theta^{\prime} \in \Theta$,
    \item $\vert \Theta \vert$ is the size of the set $\Theta$, typically 8.
\end{enumerate}
\label{thm:weak}
\end{theorem}
The ideal resource \emph{random RSP for blindness} is specified as follows. If the server is honest, the functionality outputs $\ketbra{+_{\theta}}{+_{\theta}}$ to the server. 
If not, it takes as input from the server the classical description of a quantum state $[\rho^{\theta}]$ and outputs the corresponding quantum state $\rho^{\theta}$ to the server. In both cases, the client receives the classical angle $\theta$. This is formalized in the following definition, and depicted in Figure~\ref{fig:RSP}(a))
\begin{definition}
 The ideal resource \emph{random remote state preparation for blindness}, denoted $RSP_{B}$, has two interfaces, $\mathcal{A}$ to the client and $\mathcal{B}$ to the server. The resource chooses an angle of rotation $\theta$ uniformly at random from the set $\mathbb{Z}_{\frac{\pi}{4}}$. There is a filtered functionality at interface $\mathcal{B}$ and a classical bit c. If $c=0$, the server is honest and the resource outputs a state $\ketbra{+_{\theta}}{+_{\theta}}$ on $\mathcal{B}$. If $c=1$, the ideal functionality takes as input the set \(\{(\theta,[\rho^{\theta}])\}_{\theta}\) from the server.
 %where $[\rho^{\theta}]$ denotes the classical description of a quantum state. 
 
 If the states provided by the server do not satisfy the conditions from Theorem~\ref{thm:weak}, $RSP_{B}$ ignores the input and waits for a new valid set.
 Once the set is received, the functionality outputs $\rho^{\theta}$ at $\mathcal B$. In both case, $RSP_{B}$ outputs the angle $\theta$ at the client's interface. 
\label{defn:rspb}
\end{definition}

Dunjko and Kashefi also introduced another resource that is better suited for our purpose.
It is a
variant of $RSP_{B}$ which allows more operations for a dishonest server.
This resources is depicted in Figure~\ref{fig:RSP}(b)).
\begin{definition}
The ideal resource \emph{measurement-based remote blind state preparation} ($MRSP_{B}$) has two interfaces $\mathcal{A}$ and $\mathcal{B}$. The resource chooses an angle of rotation $\theta$ uniformly at random from the set $\mathbb{Z}_{\frac{\pi}{4}}$. There is a filtered functionality at interface 
$\mathcal{B}$ and a classical bit c. If $c=0$, the server is honest and the resource outputs a state $\ketbra{+_{\theta}}{+_{\theta}}$ on $\mathcal{B}$. If $c=1$, the ideal functionality 
takes as input
the descriptions of eight positive operators \(\{\Pi_{\theta}\}\), such that for all $\theta$ in $\mathbb{Z}_{\frac{\pi}{4}}$, \(\Pi_{\theta}+\Pi_{\theta+\pi}=I\).
In addition, it accepts an arbitrary quantum state $\rho$ of the same dimension as the operator $\Pi_{\theta}$.

If the server's input does not satisfy the properties of Theorem~\ref{thm:weak}, $MRSP_{B}$ ignores it and waits for a new valid set. 
Once a valid input is received, $MRSP_{B}$ applies the measurement 
$\Pi_{\theta},\Pi_{\theta+\pi}$ corresponding to the chosen angle $\theta$ to $\rho$.
Finally, $MRSP_{B}$ outputs the measurement result 
$\theta^{\prime}$, whose value is either $\theta$ or $\theta+\pi$, at the client's interface and
the post-measurement state $\rho^{\theta^{ \prime}}$ at the server's interface. 
\label{defn:mrspb}
\end{definition}
% \begin{figure}[h!]
%         \centering
%         \includegraphics[width=13cm]{img/RSP.png}
%         \caption{RSP ideal resources for blindness (a) and measured-based RSP for blindness (b)}
%         \label{fig:RSP} 
% \end{figure}
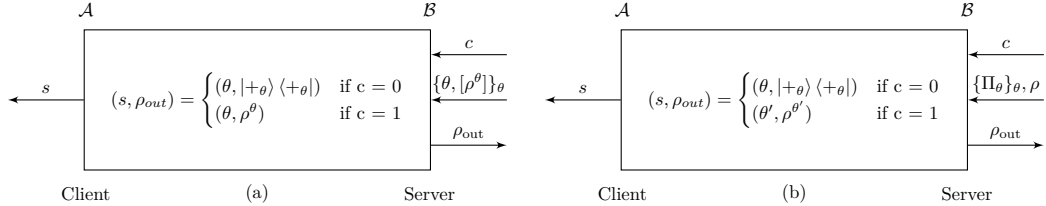
\begin{figure}[h!]
    \centering
    \begin{tikzpicture}[every node/.style = {rectangle, align=center,scale=0.7}, scale=0.7]
    \node(blindness)[draw, inner xsep=5mm, inner ysep=8mm, minimum size=1.5cm]{
    $(s,\rho_{out})=
    \begin{cases}
    (\theta, \ket{+_\theta}\bra{+_{\theta}}) &\text{if c = 0}\\
    (\theta, \rho^\theta) &\text{if c = 1}
    \end{cases}$};
    \node(veri)[draw, inner xsep=5mm, inner ysep=8mm, minimum size=1.5cm, right=2.5cm of blindness]{
    $(s,\rho_{out})=
    \begin{cases}
    (\theta, \ket{+_{\theta}}\bra{+_{\theta}}) &\text{if c = 0}\\
    (\theta^{\prime}, \rho^{\theta^\prime}) &\text{if c = 1}
    \end{cases}$};
    \node(a) [below = 1mm of blindness] {(a)};
    \node(b) [below = 1mm of veri] {(b)};
    \node(client_a) [left = 1.6cm of a] {Client};
    \node(server_a) [right = 1.6cm of a] {Server};
    \node(client_b) [left = 1.6cm of b] {Client};
    \node(server_b) [right = 1.6cm of b] {Server};
    \node(interface_a1) [above = 2.1cm of client_a] {$\A$};
    \node(interface_a2) [above = 2.1cm of server_a] {$\B$};
    \node(interface_b1) [above = 2.1cm of client_b] {$\A$};
    \node(interface_b2) [above = 2.1cm of server_b] {$\B$};
    \coordinate[left = of blindness] (a0);
    \coordinate[right = of blindness] (a3);
    \coordinate[above = 6mm of a3] (a4);
    \coordinate[below = 6mm of a3] (a5);
    \draw[-latex'] (a0-| blindness.west) -- node[pos=0.5, above] {$s$} (a0);
    \draw[-latex'] (a3) -- node[pos=0.5, above] {$\{\theta,[\rho^\theta]\}_\theta$} (a3-| blindness.east);
    \draw[-latex'] (a4) -- node[pos=0.5, above] {$c$}  (a4-| blindness.east);
    \draw[-latex'] (a5-| blindness.east) -- node[pos=0.5, above] {$\rho_\out$} (a5);
    \coordinate[left = of veri] (b0);
    \coordinate[right = of veri] (b3);
    \coordinate[above = 6mm of b3] (b4);
    \coordinate[below = 6mm of b3] (b5);
    \draw[-latex'] (b0-| veri.west) -- node[pos=0.5, above] {$s$} (b0);
    \draw[-latex'] (b4) -- node[pos=0.5, above] {$c$}  (b4-| veri.east);
    \draw[-latex'] (b3) -- node[pos=0.5, above] {$\{\Pi_\theta\}_\theta,\rho$}  (b3-| veri.east);
    \draw[-latex'] (b5-| veri.east) -- node[pos=0.5, above] {$\rho_\out$} (b5);
    \end{tikzpicture}
    \caption{RSP ideal resources for blindness (a) and measured-based RSP for blindness (b)}
    \label{fig:RSP} 
\end{figure}

The connection between these two ideal resources follows from the construction of
$MRSP_{B}$ from $RSP_{B}$, which preservers both correctness and security.
Consider a trivial protocol \(\pi=(\pi_{A}, \pi_{B})\) in which $\pi_{A}$ does nothing and $\pi_{B}$ fixes the classical bit to $c=0$. Following the conditions of Definition \ref{defn:ac}, it was shown that:
\begin{align}
    \pi_{A}RSP_{B}\pi_{B}=MRSP_{B}\perp \hspace{1cm}and\hspace{1cm} \pi_{A} RSP_{B}=MRSP_{B}\sigma_{B}. 
\end{align}
The inputs and outputs of these two ideal resources are trivially equivalent in the honest case, which implies the correctness. To prove the security, the authors provided a simulator~$\sigma_{B}$ and showed that the outputs of $\mathcal{A}$ and $\mathcal{B}$ are the same actions for $\pi_{A}RSP_{B}$ and $MRSP_{B}\sigma_{B}$. Moreover, the authors show that $RSP_{B}$ and $MRSP_{B}$
can be used for UBQC. This leads to a perfect blind DQC without a quantum channel between the client and the server. In this context, perfect blindness means that the protocol leaks nothing more than what is strictly required (such as, for example, the size of the computation). The formal definition can be found in~\cite{dunjko_blind_2016}.

The following theorem formalizes this argument for $MRSP_{B}$.
\begin{theorem}{~\cite{dunjko_blind_2016}}
The UBQC protocol in which the client has access to the ideal functionality $MRSP_{B}$ 
rather than to a quantum channel and a random generator of the $\ket{+_{\theta}}$ states,
exactly constructs DQC with perfect blindness.
\label{thm:MRSPDQC}
\end{theorem}

\subsection{Limitations of RSP with only Classical Channel}
While $RSP_B$ and $MRSP_B$ remove the need for a quantum communication channel between the client and the server, we have not discussed how these resources can be implemented.
A fully-classical blind DQC protocol could be obtained for example by implementing one of the two resources with a classical communication channel.
This idea was investigated by~\cite{badertscher_security_2020}, who introduced the following definition.
\begin{definition}
An ideal resource $\mathcal{S}$ is said to be $\varepsilon$-classical-realizable if it is realizable from a classical channel $\mathcal{C}$, i.e. if there exists a protocol $\pi=(\pi_{A}, \pi_{B})$ between two parties interacting classically such that:
\begin{align}
    \mathcal{C}\xrightarrow{\pi,\varepsilon}\mathcal{S}
\end{align}
\label{defn:realizable}
\end{definition}
In order to prove the composable security of $\varepsilon$-classical-realizable RSP,
we need to show that no unbounded adversary can learn information on $\theta$
by accessing only the right interface $\mathcal{B}$. Unfortunately, the authors show that there is no \emph{describable} remote state preparation protocol with composable security. In this context, describable means extracting a classical approximate description of a quantum state $[\rho]$ by accessing the state $\rho$ on the interface $\mathcal{B}$.
Since a protocol using only classical communication is obviously describable, 
it follows that there is no classical-realizable RSP with composable security. This in turn implies that UBQC with classical-realizable RSP cannot be composable secure. 

As a result, it is necessary to make additional assumptions to remove the quantum interaction between the client and the server. While~\cite{badertscher_security_2020} considers additional computational assumption to bound the adversary's power, we take a different approach, introducing additional hardware assumptions, such
as tamper-proof quantum operations, in order to get a secure DQC protocol with blindness using only classical communication. 

\section{QEnclave as an Ideal Functionality: Remote State Rotation}
\label{sec:RSR}
In this section, we introduce an alternative to RSP named \emph{remote state rotation} (RSR), and analyze its composable security in the AC framework. Compared to the other ideal functionalities, RSR is even weaker. While RSP generates quantum states by itself, RSR only allows rotations of single-qubit states generated by the server.

%\subsection{Ideal Functionality: Remote State Rotation}
%We first introducing the ideal functionality for remote state rotation. As we showed in the previous section, it is sufficient to have weak correlations for blindness. A formal definition is given as follows:
\begin{definition}
\textnormal{(See Figure \ref{fig:RSR})} The ideal resource named \emph{remote state rotation} for blindness ($RSR_{B}$) has two interfaces $\mathcal{A}$ and $\mathcal{B}$. 
After receiving a single-qubit state $\rho_{in}$  from interface $\mathcal{B}$, it performs a 
rotation $Z(\theta)$ with $\theta$ chosen uniformly at random from the set $\mathbb{Z}_{\frac{\pi}{4}}$. It then outputs ($\rho_{out}$) at the server's interface and the angle $\theta$ at the client's interface. 
\label{defn:rsr}
\end{definition}
% \begin{figure}[h!]
%         \centering
%         \includegraphics[width=8cm, scale=0.4]{img/RSR.png}
%         \caption{Ideal Functionality of QEnclave: Remote State Rotation}
%         \label{fig:RSR} 
% \end{figure}
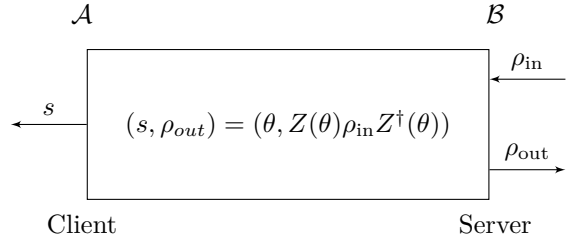
\begin{figure}[h!]
    \centering
    \begin{tikzpicture}[every node/.style = {rectangle, align=center}]
    \node(blindness)[draw, inner xsep=5mm, inner ysep=8mm, minimum size=1.5cm]{
    $(s,\rho_{out})=(\theta,Z(\theta)\rho_\inp Z^\dagger(\theta))$};
    \node(a) [below = 2mm of blindness] {};
    \node(client_a) [left = 2cm of a] {Client};
    \node(server_a) [right = 2cm of a] {Server};
    \node(interface_a1) [above = 2.3cm of client_a] {$\A$};
    \node(interface_a2) [above = 2.3cm of server_a] {$\B$};
    \coordinate[left = of blindness] (a0);
    \coordinate[right = of blindness] (a3);
    \coordinate[above = 6mm of a3] (a4);
    \coordinate[below = 6mm of a3] (a5);
    \draw[-latex'] (a0-| blindness.west) -- node[pos=0.5, above] {$s$} (a0);
    \draw[-latex'] (a4) -- node[pos=0.5, above] {$\rho_\inp$}  (a4-| blindness.east);
    \draw[-latex'] (a5-| blindness.east) -- node[pos=0.5, above] {$\rho_\out$} (a5);
    \end{tikzpicture}
    \caption{Ideal Functionality of QEnclave: Remote State Rotation}
    \label{fig:RSR}
\end{figure}

Similar to $RSP_B$ and $MRSP_B$, this functionality fully removes any quantum capability for the client. In particular, using $RSR_B$ removes the assumption for a quantum communication channel between the client and the server.

We further define a two party protocol \(\pi=(\pi_{A}, \pi_{B})\) to prepare quantum states with $RSR_{B}$ in which $\pi_{A}$ only receives the angle $\theta$ from the interface $\mathcal{A}$ of $RSR_{B}$, and $\pi_{B}$ takes as input a classical bit $c$ and a quantum state from the server.

If $c=0$, the server is honest and $\pi_{B}$ accepts $\ketbra{+}{+}$ as input from the server. If $c=1$, the dishonest server prepares an arbitrary quantum state \(\rho=\Omega(\ket{+}\bra{+}\otimes\rho_{aux})\Omega^{\dagger}\). Here, $\Omega$ is an arbitrary unitary that represents the server's deviation on the quantum source and $\rho_{aux}$ is an auxiliary state chosen by the server.

After tracing out the auxiliary state of $\rho$, we get $\rho_{in}$, the input to $RSR_{B}$ which is a single-qubit chosen by the dishonest server. In particular, this state can be entangled with the server's auxiliary system.

We show a construction of DQC with $RSR_{B}$ that achieves perfect blindness in two steps. First, In Lemma~\ref{lemma:correlation}, we prove that the outcome of $RSR_{B}$ satisfies the conditions for blindness of Theorem~\ref{thm:weak}. Then, in Theorems~\ref{thm:indistinguishability} and~\ref{thm:RSRUBQC}, we show the security of DQC with blindness obtained from $RSR_{B}$.

\begin{lemma}
For any quantum states $\rho_{in}$ as used as input of $RSR_B$, the outcome system of the client and the server $\sigma_{AB}$ satisfies the conditions of weak correlation of UBQC.
\label{lemma:correlation}
\end{lemma}

\begin{proof}
For simplicity, we first assume that $\rho_{in}$ is not entangled with the server's auxiliary system. Without
loss of generality,
we get \(\rho_{in}=|\alpha|^{2}\ket{0}\bra{0}+\alpha\beta^\ast\ket{0}\bra{1}+\alpha^\ast\beta\ket{1}\bra{0}+|\beta|^{2}\ket{1}\bra{1}\).
In this case, the output of $RSR_{B}$ $\rho^{\theta}$ is
\begin{align}
    \rho^{\theta}=|\alpha|^{2}\ket{0}\bra{0}+e^{-i\theta}\alpha\beta^\ast\ket{0}\bra{1}+e^{i\theta}\alpha^\ast\beta\ket{1}\bra{0}+|\beta|^{2}\ket{1}\bra{1}.
\label{eqn:singlequbit}
\end{align}
For any $\theta$ in the set $\mathbb{Z}_{\frac{\pi}{4}}$, we thus have
\begin{align}
    \rho^{\theta}+\rho^{\theta+\pi}&=2|\alpha|^{2}\ket{0}\bra{0}+2|\beta|^{2}\ket{1}\bra{1}.
\end{align}
Since this is independent of $\theta$, the state satisfies the weak correlation conditions.

In the general case, $\rho_{in}$ can be entangled with the server's auxiliary system.
we thus write
\(\rho_{in}^{\prime}=|\alpha|^{2}\ket{0}\bra{0}\otimes\ket{\psi_{0}}\bra{\psi_{0}}+\alpha\beta^\ast\ket{0}\bra{1}\otimes\ket{\psi_{0}}\bra{\psi_{1}}+\alpha^\ast\beta\ket{1}\bra{0}\otimes\ket{\psi_{1}}\bra{\psi_{0}}+|\beta|^{2}\ket{1}\bra{1}\otimes\ket{\psi_{1}}\bra{\psi_{1}}\), where $\ket{\psi_{0}}$ and $\ket{\psi_{1}}$ are states of the server's auxiliary system.

After the rotation of $RSR_{B}$ on the first subsystem, we get the following entangled state:
\begin{multline}
    \rho^{\theta}=|\alpha|^{2}\ket{0}\bra{0}\otimes\ket{\psi_{0}}\bra{\psi_{0}}+e^{-i\theta}\alpha\beta^\ast\ket{0}\bra{1}\otimes\ket{\psi_{0}}\bra{\psi_{1}}\\
    +e^{i\theta}\alpha^\ast\beta\ket{1}\bra{0}\otimes\ket{\psi_{1}}\bra{\psi_{0}}+|\beta|^{2}\ket{1}\bra{1}\otimes\ket{\psi_{1}}\bra{\psi_{1}}.
\label{eqn:entangledqubit}
\end{multline}
For any $\theta$ in the set $\mathbb{Z}_{\frac{\pi}{4}}$, we have
\begin{align}
    \rho^{\theta}+\rho^{\theta+\pi}&=2|\alpha|^{2}\ket{0}\bra{0}\otimes \ket{\psi_{0}}\bra{\psi_{0}}+2|\beta|^{2}\ket{1}\bra{1}\otimes\ket{\psi_{1}}\bra{\psi_{1}}.
\end{align}
Since the result $\rho^{\theta}+\rho^{\theta+\pi}$ is again, independent of $\theta$, the joint state of the client and the server also satisfy the weak correlation conditions
for any state $\sigma_{AB}$.
\end{proof}

We now prove the security of $RSR_B$ with the UBQC protocol. We prove it by showing that the resource $MRSP_{B}$ introduced in Definition \ref{defn:mrspb} can be constructed from $RSR_{B}$. Since $MRSP_{B}$ can be composed with a UBQC protocol to get DQC with perfect blindness, so does $RSR_{B}$

\begin{theorem}
The protocol \(\pi=(\pi_{A}, \pi_{B})\) introduced above with ideal resource $RSR_{B}$ constructs the ideal resource $MRSP_{B}$. 
\label{thm:indistinguishability}
\end{theorem}
\begin{proof}
We show that both the correctness and the security condition are satisfied. 
More precisely, proving the security amounts to showing that a distinguisher cannot distinguish $MRSP_{B}$ from the protocol. This translates into the following equations, for a simulator $\sigma_B$ and the protocol \(\pi=(\pi_{A}, \pi_{B})\) with $RSR_B$. 
\begin{align}
\label{eqn:correctness}
    \pi_{A}RSR_{B}\pi_{B}\approx_{\varepsilon}MRSP_{B}\perp,
\end{align}
and
\begin{align}
\label{eqn:security}
    \pi_{A}RSR_{B}\approx_{\varepsilon}MRSP_{B}\sigma_{B}.
\end{align}
For the correctness, when the server is honest, the  ideal resources $RSR_{B}$ and $MRSP_{B}$ both output an angle $\theta$ at interface $\mathcal{A}$ and its corresponding quantum state $\ket{+_{\theta}}\bra{+_{\theta}}$ at interface $\mathcal{B}$. Equation~\ref{eqn:correctness} is thus immediately satisfied.

For the security, we introduce the simulator $\sigma_{B}$, defined as follows.
It accepts and sends $c=1$ to $MRSP_{B}$, as well as a set of projectors \(\{\Pi^{\theta}\}\), where \(\Pi^{\theta}=\ket{+_{-\theta}}\bra{+_{-\theta}}\). 
After receiving a quantum system $\rho$ from the server, the simulator takes the input $\rho_{in}$ of the same dimension as \(\{\Pi^{\theta}\}\), and generates a qubit $\ket{0}$.
A CNOT gate is applied to these two qubits, where $\rho_{in}$ is used as the control qubit ($\ket{\phi_{1}}$) and $\ket{0}$ the target bit ($\ket{\phi_{2}}$). 
This gives the simulator state (\(\rho_{\sigma_{B}}=\ket{\phi_{12}}\bra{\phi_{12}}\)).
Finally, $\sigma_B"$ sends the first qubit $\ket{\phi_{1}}$  back as the outcome state to the server, whereas the second qubit,
$\ket{\phi_{2}}$, is sent to the resource $MRSP_{B}$.

We show that the outcome is similar to the expression obtained in Lemma~\ref{lemma:correlation}. 
Again, consider first the case where $\rho_{in}$ is not entangled with the server's auxiliary system.
We then obtain the following expression for $\ket{\phi_{12}^{\prime}}$ after the operation of $MRSP_{B}$:
\begin{align*}
    \ket{\phi_{12}^{\prime}}&=\frac{\Pi^{\theta}_{2}}{\sqrt{\bra{\phi_{12}}\Pi^{\theta}_{2}\ket{\phi_{12}}}}\ket{\phi_{12}}\\
    &=\frac{1}{\sqrt{2}}(\ket{0}+e^{-i\theta}\ket{1})(\bra{0}+e^{i\theta}\bra{1})(\alpha\ket{00}+\beta\ket{11})\\
    % &=\frac{1}{2}(\ket{0}+e^{-i\theta}\ket{1})(\alpha\ket{0}+e^{i\theta}\beta\ket{1})\\
    &=\frac{1}{\sqrt{2}}(\alpha\ket{00}+e^{-i\theta}\alpha\ket{01}+e^{i\theta}\beta\ket{10}+\beta\ket{11})
\end{align*}
We obtain the outcome of simulator by tracing out the second quantum subsystem.
\begin{align*}
    \rho_{1}&=Tr_{2}(\ket{\phi_{12}^{\prime}}\bra{\phi_{12}^{\prime}})\\
    &=|\alpha|^{2}\ket{0}\bra{0}+e^{-i\theta}\alpha\beta^\ast\ket{0}\bra{1}+e^{i\theta}\alpha^\ast\beta\ket{1}\bra{0}+|\beta|^{2}\ket{1}\bra{1}
    % &=(1+e^{-i\theta})(\alpha\ket{0}+e^{i\theta}\beta\ket{1})
\end{align*}
The outcome quantum state is exactly the same as result as the outcome of $RSR_{B}$ in Eq.(\ref{eqn:singlequbit}). Since a similar calculation holds for the projector \(\Pi^{\theta+\pi}\),
the outcome joint state of the client and the server of $MRSP_{B}$ is exactly the same as $RSR_{B}$. 

Consider now an arbitrary entangled state \(\alpha\ket{0}\ket{\psi_{0}}+\beta\ket{1}\ket{\psi_{1}}\).
The simulator $\sigma_{B}$ takes the first single-qubit subsystem as control qubit, and performs
the same operation as in the previous case. After the operation of $MRSP_{B}$, we have:
\begin{align*}
    \ket{\phi_{12}^{\prime}}&=\frac{\Pi^{\theta}_{2}}{\sqrt{\bra{\phi_{12}}\Pi^{\theta}_{2}\ket{\phi_{12}}}}\ket{\phi_{12}}\\
    &=\frac{1}{\sqrt{2}}(\ket{0}+e^{-i\theta}\ket{1})(\bra{0}+e^{i\theta}\bra{1})(\alpha\ket{0}\ket{\psi_{0}}\ket{0}+\beta\ket{1}\ket{\psi_{1}}\ket{1})\\
    % &=\frac{1}{2}(\ket{0}+e^{-i\theta}\ket{1})(\alpha\ket{0}\ket{\psi_{0}}+e^{i\theta}\beta\ket{1}\ket{\psi_{1}})\\
    &=\frac{1}{\sqrt{2}}(\alpha\ket{0}\ket{\psi_{0}}\ket{0}+e^{-i\theta}\alpha\ket{0}\ket{\psi_{0}}\ket{1}+e^{i\theta}\beta\ket{1}\ket{\psi_{1}}\ket{0}+\beta\ket{1}\ket{\psi_{1}}\ket{1})
\end{align*}
Then, after tracing out the second qubit, we obtain:
\begin{align*}
    \rho_{1}=&Tr_{2}(\ket{\phi_{12}^{\prime}}\bra{\phi_{12}^{\prime}})\\
    =&|\alpha|^{2}\ket{0}\bra{0}\otimes\ket{\psi_{0}}\bra{\psi_{0}}+e^{-i\theta}\alpha\beta^\ast\ket{0}\bra{1}\otimes\ket{\psi_{0}}\bra{\psi_{1}}\\
    & +e^{i\theta}\alpha^\ast\beta\ket{1}\bra{0}\otimes\ket{\psi_{1}}\bra{\psi_{0}}+|\beta|^{2}\ket{1}\bra{1}\otimes\ket{\psi_{1}}\bra{\psi_{1}}.
    % &=(1+e^{-i\theta})(\alpha\ket{0}+e^{i\theta}\beta\ket{1})
\end{align*}
Again, The output quantum state is exactly equals to $\rho^{\theta}$ specified in Eq.(\ref{eqn:entangledqubit}). In consequence, the resource $RSR_{B}$ is perfectly indistinguishable with the resource $MRSP_{B}$, that is, Equations~\ref{eqn:correctness} and~\ref{eqn:security} are satisfied
with $\epsilon=0$.

\end{proof}

Finally, combining the fact that we can perfectly construct $MRSP_{B}$ from $RSR_{B}$
with Theorem~\ref{thm:MRSPDQC}, we obtain the following result.

\begin{theorem}
The UBQC protocol with the client accessing to the $RSR_{B}$ constructs the ideal functionality of DQC with perfect blindness.
\label{thm:RSRUBQC}
\end{theorem}

\section{Specification of the QEnclave}
\label{sec:QEnclave}
In this section, we give a complete specification of the QEnclave, based on the use of a secure processor. The QEnclave implements the ideal functionality RSR, and uses the enclave and its hardware assumption to ensure the security of the construction.
Moreover, it communicates with the client classically and returns a quantum state to the server, as
shown in Figure~\ref{fig:SS}.
For convenience, we assume that the client can choose the input angles uniformly at random,
rather than letting the QEnclave choose them (as in $RSR_B$). 
This transformation does not change the security 
since in our setup, both the client and QEnclave are expected to be honest.
\begin{figure}[h!]
        \centering
        \includegraphics[width=10cm]{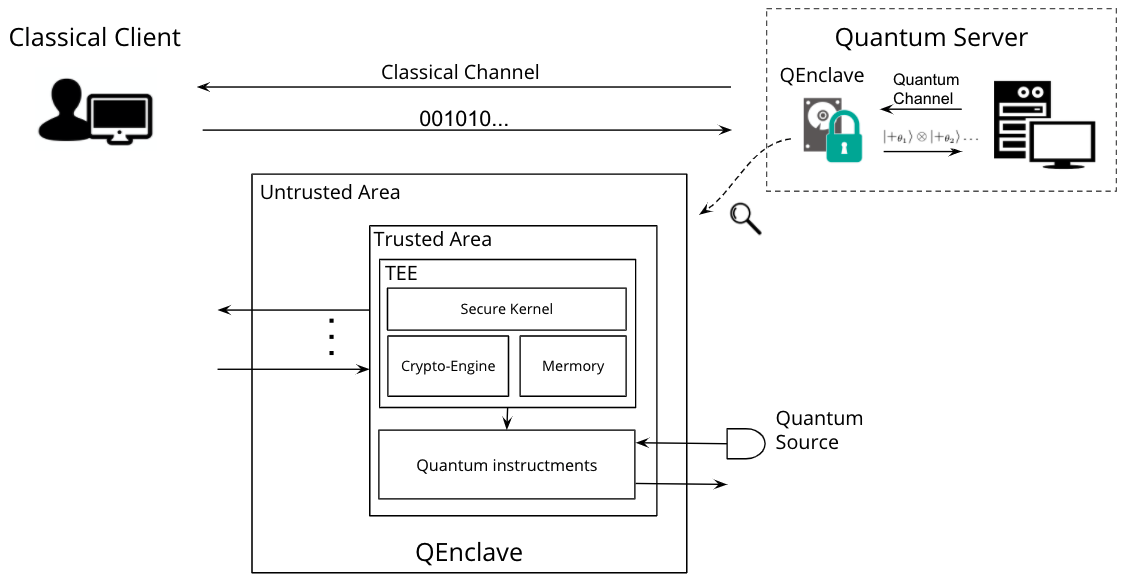}
        \caption{Specification of QEnclave}
        \label{fig:SS} 
\end{figure}

To start with, we demonstrate how attested execution functionality $\G_\att$ securely constructs outsourcing computation protocol under composition. By considering a simple 2-party outsourcing computation $\F_{\text{outsrc}}[\C,\Ss]$ with target function $y=f(x)$, where the client $\C$ outsources $f$ and $x$ with encoding and finally obtains the output $y$ while the server $\Ss$ or any other adversary only knows the size of inputs and outputs $(\abs{f+x}, \abs{y})$ during the computation process. A $\G_\att$-hybrid protocol $\texttt{Prot}_\text{outsrc}$ (Protocol \ref{protocol:outsrc}) is given in~\cite{coron_formal_2017} and proven to UC-realized $\F_{\text{outsrc}}$ when $\C$ is honest and $\Ss$ is a static adversary (See Appendix \ref{sec:review}). The probabilistic polynomial time indistinguishability of ideal-world and real-world executions is reduced to the Decisional Diffie-Hellman (DDH) assumption for secure key exchange~\cite{Maurer2013KeyEW} and authenticated encryption. The indistinguishability is also equivalent under AC framework without instantiating a DDH-based secure key exchange protocol but a random secret key shared between $\C$ and $\G_\att$. Furthermore, since the composable security of $RSR$, the construction of QEnclave with secure processor's functionality $\G_\att$ is theoretically feasible.

% \begin{corollary}
% For $\G_\att$-hybrid protocol $\texttt{Prot}_\text{outsrc}$ with enclave program $\prog_\text{outsrc}$, 
% \end{corollary}

% \yao{(To be modified)} One immediate question is the construction of a secure classical communication channel between the client and QEnclave. This means that the communication should be encrypted and authenticated, in particular to protect it from server's attacks.
% This requires a persistent tamper-proof memory for device authentication and an isolated and tamper-proof crypto-engine for generating and sharing secret keys in the secure processor of QEnclave. 

% In practice, the honest client should authenticate the status of QEnclave before sending the secret. This process is addressed by the remote attestation scheme we introduce previously. Notice that a quantum-safe digital signature scheme~\cite{akleylek_efficient_2016,barreto_sharper_2016,yang_xmss_2011} is necessary for the remote attestation scheme since we assume that the quantum server is potentially dishonest. Fortunately, there exists remote attestation schemes providing quantum-proof authentication~\cite{liu_remote_2018}.

In practice, a quantum-safe digital signature scheme~\cite{akleylek_efficient_2016,yang_xmss_2011} is necessary for the remote attestation scheme since we assume that the quantum server is potentially malicious. Meanwhile, more practical remote attestation schemes provide post-quantum security~\cite{liu_remote_2018}.

The confidentiality consists in hiding the rotation angles chosen by the client. The requirement of using a quantum-safe encryption makes symmetric schemes more appropriate than asymmetric ones for this task. Apart from a key exchange protocol based on DDH, there are also other key encapsulation mechanism (KEM) schemes~\cite{baldi_ledakem_2018, bindel_hybrid_2018,wang_post-quantum_2020} theoretically feasible to share a secret key between the client and the QEnclave. 
% More precisely, it requires the client to generate the asymmetric keys and send its public key to QEnclave. The symmetric secret key is then generated by the QEnclave and stored inside the tamper-proof memory of the secure processor. The client's public key is used to encrypt the secret symmetric key and the ciphertext is returned to the client.
%With the private key, the client acquires the symmetric key from QEnclave exclusively. 

Once the secure channel is established  between the client and QEnclave, the client can send the encrypted rotation angles to the QEnclave. QEnclave decrypts them and encodes the initial quantum state from the external source using the classical angles chosen by the client. At this stage, we assume that the trusted area which includes the secure processor also includes the quantum apparatus and protects the interactions between the two from external eavesdropping. This leads to a remote state preparation protocol for delegated quantum computation with blindness using the QEnclave and classical communication between the client and the server.
We summarize all the Steps in Protocol~\ref{protcol:actual}.
\begin{protocol}
\caption{QEnclave-based RSP Protocol for Blindness with $\prog_\text{rsr}$}
\vspace{.2cm}
$\G_\att$-enabled QEnclave Program $\prog_\text{rsr}$:
\vspace{.2cm}

\quad On input $(\text{``keyex''},g^a):$ 

\quad\quad let $b\in\mathbb{Z}_p$, store $\text{sk}:=(g^a)^b$; return $(g^a,g^b)$

\quad On input* $(\text{``compute''},ct)$:

\quad\quad let $(f,x):=\text{AE.Dec}_{\text{sk}}(ct)$
    
\quad\quad assert decryption success, $ct$ not seen before
    
\quad\quad let $\theta_0\dots\theta_n:=f(x)$ and return $ct_\out:=\text{AE.Enc}_{\text{sk}}(\theta_0\dots\theta_n)$,

\quad\quad while $\theta_0\dots\theta_n$ is applied to quantum states from external quantum source:
\begin{align*}
        \ket{e}=Z_{1}(\theta_{1})\otimes\ldots\otimes Z_{n}(\theta_{n})
        \ket{+}^{\otimes n}
\end{align*}
\vspace{.2cm}
Server $\Ss$:
\vspace{.2cm}

\quad On receive $(\text{``keyex''},g^a)$ from $\C$:
    
\quad\quad let $\emph{eid}:=\G_\att.\install(\emph{sid},\prog_\text{rsr})$
    
\quad\quad let $((g^a,g^b),\sigma):=\G_\att.\resume(\emph{eid},(\text{``keyex''},g^a))$ and send $(\emph{eid},g^b,\sigma)$ to $\C$
    
\quad On receive* $(\text{``compute''},ct)$ from $\C$:
    
\quad\quad let $(ct_\out,\sigma):=\G_\att.\resume(\emph{eid},(\text{``compute''},ct)$ and send $ct_\out$ to $\C$

\vspace{.2cm}
Client $\C$:
\vspace{.2cm}

\quad On initialize:
    
\quad\quad let $a\in \mathbb{Z}_p$,$\mpk:=\G_\att.\getpk()$
    
\quad\quad send $(\text{``keyex''},g^a)$ to $\Ss$, await $(\emph{eid}, g^b, \sigma)$ from $\Ss$
    
\quad\quad assert $\Sigma.\text{Vf}_\mpk((\emph{sid},\emph{eid},\texttt{Prot}_\text{rsr},(g^a,g^b)),\sigma)$
    
\quad\quad let $\text{sk}:=(g^b)^a$
    
\quad On receive* $(\text{``compute''},f,x)$:
    
\quad\quad let $ct:=\text{AE.Enc}_\text{sk}(f,x)$ and send $(\text{``compute''},ct)$ to $\Ss$, await $ct_\out$
    
\quad\quad let $y:=\text{AE.Dec}_\text{sk}(ct_\out)$ and assert decryption success and $ct_\out$ not seen before output $y$

\label{protcol:actual}
\end{protocol}

Notice that this protocol assumes a reliable quantum apparatus inside the trusted area of QEnclave for transforming the incoming quantum state. Moreover, we assume that the communication between the secure processor and the quantum apparatus remains hidden to the server.
While this assumption may seem strong, the idea of sealing hardware components into a tamper-proof box is already well spread in the world of hardware security. In particular,
the FIPS-140 certification for Hardware Security Modules (HSM) includes criteria
for physical tamper-evidence (level 2 certification), physical tamper-resistance (level 3) or even robustness against environmental attacks (level 4).

While our proposal of QEnclave implements the RSR ideal functionality, we have left aside a number of potential attacks that stem from the physical realization.
Implicitly, we assume that QEnclave is fabricated correctly by a certified manufacturer, which ensures that an adversary cannot subvert the device before it was installed on the server. Besides, we also exclude some hardware-dependent attacks in our work \emph{e.g.}, specific side-channel attacks on specific enclave products. Finally, we have not yet considered the possibility of counterfeiting the QEnclave. 

\section{Conclusion and Discussion}
\label{sec:concl}
We introduced a new functionality called \emph{Remote State Rotation}, that can be used to achieve secure delegated quantum computing in a practical way compatible with currently available quantum hardware platform in the cloud. Moreover, we have proposed a realistic hardware assumption of trustworthy quantum operations with classical secrets to circumvent the impossibility results of~\cite{Aaronson2019ComplexityTheoreticLO,badertscher_security_2020} of implementing a composable RSR with classical channel only. Our proposed ideal functionality with simple rotations lowers the minimal requirement on the operations of the client while keeping minimal overhead on the server side too. Finally, we gave a full specification of QEnclave that implements the RSR functionality using a secure processor to control the quantum apparatus required for the blindness of delegated quantum computation.

Beside privacy, another desirable property of delegated quantum computing is verifiability. In general, a DQC protocol is verifiable if the client can verify the result from the server (See Definition~\ref{defn: VBQC}).
A verifiable universal blind quantum computing protocol was proposed in~\cite{fitzsimons_unconditionally_2017} where the client could inserts in the target computation a set of trap qubits that are isolated from the computation. This construction ensures that the measurement results of trap qubits is always deterministic and can be used as a test of the correctness of the entire computation as they are known only by the client. 

Adapting the same approach for RSR is not trivial as a malicious server controlling the source, is now enabled to perform correlated attack before and after the call to RSR. Hence the proof technique from~\cite{fitzsimons_unconditionally_2017} does not directly apply. In principle, such deviations can be chosen to affect certain types of computation qubits but leave trap qubits unchanged, then change the execution of the protocol but remain unnoticed by the client at the same time, which means the protocol is not verifiable. However there exist many other approaches to verifiability such as the ones based on self-testing that might prove more suitable for RSR. We leave this question open for the future work. It is worth mentioning that one could trivially add a trusted measurement device or a trusted source to the construction of the QEnclave to remove the possibility of such correlated attack implementing directly the RSP resource instead. This will define directly an efficient classical client verifiable delegated computing protocol with extended hardware assumption addressing the current challenge of demonstrating certifiable quantum supremacy. However we believe keeping the QEnclave construction as simple as possible is a more interesting option to be explored. 

Finally one can explore usability of QEnclave for any quantum protocols that can be implemented through RSR. In particular, we think it can be relevant to use it to for a practical implementation of semi-quantum money schemes~\cite{radian_semi-quantum_2019} that unlike standard quantum money protocol~\cite{wiesner_conjugate_1983} considers that the bank mints the quantum states used as banknotes on the user's side and verifies their validity using only classical interactions. This matches our definition of remote state preparation, once the problem of verifiability is also addressed. Then using an QEnclave, a bank might be able to authenticate the banknote by remotely performing quantum operations but using only classical communication.

%%%% 8. BILBIOGRAPHY %%%%
\newpage
\bibliographystyle{plainurl}
\bibliography{abbrev3,crypto,biblio}
%%%% NOTES
% - Download abbrev3.bib and crypto.bib from https://cryptobib.di.ens.fr/
% - Use bilbio.bib for additional references not in the cryptobib database.
%   If possible, take them from DBLP.

\newpage
\appendix
\section{$\G_\att$-hybrid Protocol with $\prog_\text{outsrc}$}
\label{sec:review}
In this section, we review the composable security proof while using $\G_\att$ as hardware assumption to construct a outsourcing protocol \ref{protocol:outsrc} as real world execution in compared with the ideal functionality $\F_{\text{outsrc}}$. In the protocol, the server equips with a secure processor with $\G_\att$ functionality and the enclave program $\prog_\text{outsrc}$ running within it. 
\begin{protocol}[!h]
\caption{$\G_\att$-based Outsourcing Protocol $\texttt{Prot}_\text{outsrc}[\emph{sid},\C,\Ss]$~\cite{coron_formal_2017}}

\vspace{.2cm}
$\G_\att$-enabled Program $\prog_\text{outsrc}$:
\vspace{.2cm}

\quad On input $(\text{``keyex''},g^a):$ 

\quad\quad let $b\in\mathbb{Z}_p$, store $\text{sk}:=(g^a)^b$; return $(g^a,g^b)$

\quad On input* $(\text{``compute''},ct)$:

\quad\quad let $(f,x):=\text{AE.Dec}_{\text{sk}}(ct)$
    
\quad\quad assert decryption success, $ct$ not seen before
    
\quad\quad let $y:=f(x)$ and return $ct_\out:=\text{AE.Enc}_{\text{sk}}(y)$

\vspace{.2cm}
Server $\Ss$:
\vspace{.2cm}

\quad On receive $(\text{``keyex''},g^a)$ from $\C$:
    
\quad\quad let $\emph{eid}:=\G_\att.\install(\emph{sid},\prog_\text{outsrc})$
    
\quad\quad let $((g^a,g^b),\sigma):=\G_\att.\resume(\emph{eid},(\text{``keyex''},g^a))$ and send $(\emph{eid},g^b,\sigma)$ to $\C$
    
\quad On receive* $(\text{``compute''},ct)$ from $\C$:
    
\quad\quad let $(ct_\out,\sigma):=\G_\att.\resume(\emph{eid},(\text{``compute''},ct)$ and send $ct_\out$ to $\C$

\vspace{.2cm}
Client $\C$:
\vspace{.2cm}

\quad On initialize:
    
\quad\quad let $a\in \mathbb{Z}_p$,$\mpk:=\G_\att.\getpk()$
    
\quad\quad send $(\text{``keyex''},g^a)$ to $\Ss$, await $(\emph{eid}, g^b, \sigma)$ from $\Ss$
    
\quad\quad assert $\Sigma.\text{Vf}_\mpk((\emph{sid},\emph{eid},\texttt{Prot}_\text{outsrc},(g^a,g^b)),\sigma)$
    
\quad\quad let $\text{sk}:=(g^b)^a$
    
\quad On receive* $(\text{``compute''},f,x)$ from environment $\Z$:
    
\quad\quad let $ct:=\text{AE.Enc}_\text{sk}(f,x)$ and send $(\text{``compute''},ct)$ to $\Ss$, await $ct_\out$
    
\quad\quad let $y:=\text{AE.Dec}_\text{sk}(ct_\out)$ and assert decryption success and $ct_\out$ not seen before output $y$
\label{protocol:outsrc}
\end{protocol}
\begin{theorem}[Secure outsourcing from $\G_\att$~\cite{coron_formal_2017}]
Assume that the signature scheme $\Sigma$ is existentially unforgeable under chosen message attacks, the Decisional Diffie-Hellman assumption holds in the algebraic group adopted, the authenticated encryption scheme AE is perfectly correct and satisfies the standard notions of INT-CTXT and semantic security. Then, the $\G_\att$-hybrid protocol $\texttt{Prot}_\text{outsrc}$ UC-realizes $\F_{\text{outsrc}}$  when the client $\C$ is honest, and the server is a static, malicious adversary. 
\label{thm:uc_outsrc}
\end{theorem}
\begin{proof}
For the case that both $\C$ and $\Ss$ are honest is trivial since all communication between $\C$ and $\Ss$ is assumed to occur over secure channel; For the case of honest client and corrupted server, an ideal world simulator $\texttt{Sim}$ that makes no probabilistic polynomial time indistinguishability of ideal-world and real-world executions can be described as follows:
\begin{itemize}
    \item The simulator $\texttt{Sim}$ forwards any communication between $\G_\att$ and adversary (the corrupted $\Ss$) or between $\C$ and $\Ss$.
    \item $\texttt{Sim}$ starts by emulating the setup of a secure channel between $\C$ and $\G_\att$. $\texttt{Sim}$ sends $(\text{``keyex''},g^a)$ to $\Ss$ for a randomly chosen $a$.
    \item When $\texttt{Sim}$ receives a tuple $(\emph{eid},g^b,\sigma)$, $\texttt{Sim}$ aborts outputting $\texttt{sig-failure}$ if $\sigma$ would be validated by a honest $C$, while $\texttt{Sim}$ has not recorded the following communication between $\G_\att$ and $\Ss$:
    \begin{itemize}
        \item $\emph{eid}:=\G_\att.\texttt{install}(\emph{sid},\prog_\text{outsrc})$;
        \item $((g^a,g^b),\sigma):=\G_\att.\resume(\emph{eid},(\text{``keyex''},g^a))$
    \end{itemize}
    Else, $\texttt{Sim}$ computes $sk=g^{ab}$
    \item When $\texttt{Sim}$ receives a message $(\abs{f+x}, \abs{y})$ from $\F_{\text{outsrc}}$, it proceeds as follows: $\texttt{Sim}$ sends $(\text{``compute''},ct=\text{AE.Enc}_\text{sk}(f_0,x_0))$ to $\Ss$ where $f_0$ and $x_0$ are some canonical function and input. 
    \item Then, $\texttt{Sim}$ waits to receive $ct_{\out}$ from $\Ss$. If $ct_{\out}$ was not the result of a previous $\G_\att.\resume(\emph{eid},(\text{``compute''},ct)$ call but $ct_\out$ successfully decrypts under $\texttt{sk}$, the simulator aborts outputting $\texttt{authenc-failure}$. Otherwise, $\texttt{Sim}$ allows $\F_{\text{outsrc}}$ to deliver $y$ to $\C$ in the ideal world.
\end{itemize}
The indistinguishability of the ideal world and real world execution can be proven within multiple steps of hybrids:
\begin{claim}
Assume that the signature scheme $\Sigma$ is secure, except with negligible probability, the simulated execution does not abort outputting $\texttt{sig-failure}$.
\end{claim}
\begin{proof}
Straightforward reduction to the security of the digital signature scheme $\Sigma$.
\end{proof}
\begin{hybrid}
Identical to the simulated execution, but the secret key $\texttt{sk}=g^{ab}$ shared between $\C$ and $\G_\att$ is replaced with a random element from the appropriate domain.
\end{hybrid}
\begin{claim}
Assume that the DDH assumption holds, then Hybrid 1 is computationally indistinguishable from the simulated execution.
\end{claim}
\begin{proof}
Straightforward by reduction to the DDH assumption.
\end{proof}
\begin{claim}
Assume that $\texttt{AE}$ satisfies INT-CTXT security. It holds that in Hybrid 1, \newline
$\texttt{authen-failure}$ does not happen except with negligible probability.
\end{claim}
\begin{proof}
Straightforward by reduction to the INT-CTXT security of authentication encryption. If $\Ss$ makes a $\G_\att.\resume(\emph{eid},(\text{``compute''},ct^\prime)$ call where $ct^\prime$ is not the ciphertext previously sent by $\texttt{Sim}$, either $ct^\prime$ is a previously seen ciphertext (causing $\prog_\text{outsrc}$ to abort, or the decryption of $ct^\prime$ in $\prog_\text{outsrc}$ fails with overwhelming probability).

Similarly, is the output $ct_\out$ sent by $\Ss$ to $\texttt{Sim}$ does not come from a correct \newline $\G_\att.\resume(\emph{eid},(\text{``compute''},ct)$ call, then either $ct_\out$ is a previously seen ciphertext, or $\C$'s decryption would fail with overwhelming probability.
\end{proof}
\begin{hybrid}
Instead of sending $ct=\text{AE.Enc}_\text{sk}(f_0,x_0))$ to $\Ss$, the simulator now sends $ct=\text{AE.Enc}_\text{sk}(f,x))$ where $f$ and $x$ are the honest client's true inputs.
\end{hybrid}
\begin{claim}
Assume that $\texttt{AE}$ is semantically secure, Hybrid 2 is computationally indistinguishable from Hybrid 1.
\end{claim}
\begin{proof}
Straightforward reduction to the semantic security of authenticated encryption.
\end{proof}
\begin{hybrid}
Now using real key $g^{ab}$ instead of using a random key between $\C$ and $\G_\att$.
\end{hybrid}
\begin{claim}
Assume that the DDH assumption holds, then Hybrid 3 is computationally indistinguishable from Hybrid 2.
\end{claim}
\begin{proof}
Straightforward by reduction to the DDH assumption.
\end{proof}
Finally, observe that conditioned on the simulator not aborting and $\texttt{AE}$ being perfectly correct, Hybrid 3 is identically distributed as the real execution.

% For the case that both $\C$ and $\Ss$ are honest is trivial since all communication between $\C$ and $\Ss$ is assumed to occur over secure channel. However, if both parties are honest, the simulator cannot obtain valid signatures from $\G_\att$ from adversary. So if the protocol was using authenticated channels, the modeling of leaking valid signatures from $\G_\att$ to adversary is necessary, 
\end{proof}

\end{document}